\documentclass[11pt]{article}
 
 \usepackage{times}
 
 \usepackage{subcaption}
 \usepackage{graphicx}

 \def\colorful{0}
 \usepackage{times}
 \usepackage{amsfonts,amsmath,amssymb, mathtools, color,float,graphicx,verbatim}
 \usepackage[ruled,vlined,linesnumbered]{algorithm2e}
 \usepackage{hyperref} 
\usepackage{multirow}  
 \usepackage{amsthm,amsfonts,amsmath,color,float,graphicx,verbatim}
 \usepackage{multirow}
 \usepackage{amsmath, color, enumitem}
 \usepackage{framed}
 \usepackage{nicefrac}
 \usepackage{thm-restate}
 
\hypersetup{final}
 
 \usepackage{enumitem}
 \usepackage[capitalize]{cleveref}

 \def\nnewcolor{1}
 \ifnum\nnewcolor=1
 
 \fi
 \ifnum\nnewcolor=0
 
 \fi
 
 \ifnum\colorful=1

 \else

\def \E {\mathbb{E}}

\def \eps {\epsilon}

 \fi

\usepackage[dvipsnames]{xcolor}
\hypersetup{
    colorlinks,
    linkcolor={red},
    citecolor={ForestGreen},
    urlcolor={Black}
}

 \newtheorem{theorem}{Theorem}[section]
 
 \newtheorem{lemma}[theorem]{Lemma}
 
 \newtheorem{corollary}[theorem]{Corollary}

 \newtheorem{remark}[theorem]{Remark}

 \newtheorem{definition}[theorem]{Definition}

\def\E{\ensuremath{\mathrm{\mathbf{E}}}}

\def\Cov[2][]{\ensuremath{\mathrm{\mathbf{Cov}}}}

 \title{Property Testing of LP-Type Problems}

\author  {Rogers Epstein
 	\thanks{Massachusetts Institute of Technology, Cambridge, MA 02139. Email:~\url{rogersep@mit.edu}.}
 	\and Sandeep Silwal
 	\thanks{Massachusetts Institute of Technology, Cambridge, MA 02139. Email:~\url{silwal@mit.edu}. Research supported by the NSF Graduate Research Fellowship under Grant No. 1122374.}}

 \def\accept{{\fontfamily{cmss}\selectfont accept}\xspace}
 \def\reject{{\fontfamily{cmss}\selectfont reject}\xspace}

\begin{document}
 	
 	\maketitle

 	\setcounter{page}{0}
 	
 	\thispagestyle{empty}

\begin{abstract}
Given query access to a set of constraints $S$, we wish to quickly check if some objective function $\varphi$ subject to these constraints is at most a given value $k$. We approach this problem using the framework of property testing where our goal is to distinguish the case $\varphi(S) \le k$ from the case that at least an $\epsilon$ fraction of the constraints in $S$ need to be removed for $\varphi(S) \le k$ to hold. We restrict our attention to the case where $(S, \varphi)$ are LP-Type problems which is a rich family of combinatorial optimization problems with an inherent geometric structure. By utilizing a simple sampling procedure which has been used previously to study these problems, we are able to create property testers for any LP-Type problem whose query complexities are independent of the number of constraints. 
To the best of our knowledge, this is the first work that connects the area of LP-Type problems and property testing in a systematic way. Among our results is a tight upper bound on the query complexity of testing clusterability with one cluster considered by Alon, Dar, Parnas, and Ron (FOCS 2000). We also supply a corresponding tight lower bound for this problem and other LP-Type problems using geometric constructions.
\end{abstract}

\section{Introduction} \label{sec:intro}
Many problems in combinatorial optimization can be represented as a pair $\varphi(S) = (S, \varphi)$ where $S$ is a set of constraints and $\varphi$ is a function of the constraints that we would like to minimize. This class includes many problems that are NP-hard, even for the decision version of some problems where we would like to determine if $\varphi(S)$ is at most some constant $k$. For instance, let $S$ be the constraints that two nodes in a graph are connected by an edge (so $S$ can be thought of as a set of edges) and $\varphi$ be the chromatic number of a graph with those edges. Then it is NP-complete to determine if $\varphi(S) \le 3$. 

In this work, we consider a relaxation of the above hard class of problems by using the framework of \emph{property testing}. Specifically, given a value parameter $k$ and distance parameter $\epsilon$, we wish to determine if $\varphi(S) \le k$ or if $(S, \varphi)$ is $\epsilon$-far from $\varphi(S) \le k$, where $\epsilon$-far means that at least $\epsilon |S|$ many of the constraints of $S$ need to be removed for $\varphi(S) \le k$ to hold. We assume we have query access to the constraints and knowledge of $\varphi$ and our goal is to perform the property testing task by looking at a `small' number of constraints. Therefore, our query complexity is measured by how many constraints of $S$ we see. 

Even under the property testing setting, this question is too broad. For example, property testing algorithms for sparse graph problems utilize very specific combinatorial properties of the problem (see \cite{GGR, graph_prop2, graph_prop1, graph_prop3, goldreich_2017}) which means that it is not possible to study all combinatorial optimization problems together. Therefore, we focus our attention to an important class called \emph{LP-Type Problems}, which are formally described in Definition \ref{def:lptype}. Informally, these problems have an underlying geometrical structure which can be used to create efficient testing algorithms. 

Our main result is a `meta' algorithm that is able to perform property testing for \emph{any} LP-Type problem, and in many cases, is able to give tight upper bounds, such as a tight upper bound for a case of testing of clusterability considered in \cite{clustering}. To the best of our knowledge, this is the first work that connects the area of LP-Type problems and property testing in a systematic way. The class of these problems is quite general and includes problems which have been previously studied individually in property testing, such as testing clusterability of points in \cite{clustering}, and newer testing problems, such as determining if a set of linear constraints is feasible or `far' from feasible. We also give a matching lower bound for many of these problems using geometric constructions. For a comprehensive overview of our contributions, see Section \ref{sec:contributions}. 

\paragraph{Related Work}
Many problems in property testing can be modeled as a set of constraints and some optimization function over these constraints. These include well studied graph problems such as bipartiteness, expansion, $k$-colorability, and many other problems \cite{GGR, graph_prop2, graph_prop1, graph_prop3}. This line of work was initiated by Goldreich and Ron and there are many results in the area of graph property testing. The query type for these problems usually are adjacency list and adjacency matrix queries where one can ask for a particular neighbor of a vertex or if an edge exists between two given vertices. For more information about graph property testing, see \cite{goldreich_2017} and the references within. Overall, these testing problems differ from our setting where our queries are essentially access to random constraints. 

This model where queries are accesses to constraints have also been studied in the case where we wish to test properties of a metric space and queries are access to points (see \cite{metricpoints, clustering, pixeltesting}). There are instances of these problems that are also examples of LP-Type problems that we consider. For more information, see Section \ref{sec:contributions}.

LP-Type problems have a rich literature and there have been many previous work on them, including general algorithms to solve LP-Type problems \cite{Seidel1991,Clarkson:1995:LVA:201019.201036, lptype1, lptype2}. The algorithms for these problems have runtimes that are generally linear in the number of constraints, but exponential in the dimension of the LP-Type problem (see Definition \ref{def:dim}). This is in contrast to our testing algorithms that have no dependence on the number of constraints. 

Furthermore, many properties of LP-Type problems have been generalized to a larger class of problems called violator spaces \cite{violator1, violator2}. We do not explicitly consider them here since these problems do not yield any additional interesting property testing applications but our results carry over to this setting in a straightforward manner.

Lastly, there is existing work on property testing for constraint satisfaction problems (CSPs) where given an instance of a CSP, one is given query access to an assignment of the variables and the task is to determine if the assignment is `close' or `far' from satisfying the instance \cite{csp1}. This is different than our setting where we wish to check if $\varphi(S) \le k$ where $\varphi$ is a function of the constraints in $S$.

\paragraph{Organization}

In Section \ref{sec:prelim}, we formally define the class of LP-Type problems. In Section \ref{sec:contributions} we outline our contributions. In Section \ref{sec:alg}, we present our algorithms and prove their correctness and in Section \ref{sec:applications}, we apply our algorithm to specific LP-Type problems. Finally, we complement some of these problems with lower bounds in Section \ref{sec:lowerbounds}.

\section{Preliminaries} \label{sec:prelim}
\subsection{Notation and Definitions}
We formally define LP-Type problems as well as some related concepts. These definitions are standard in the literature for LP-Type problems but we reproduce them below for the sake of completeness. For more information, see \cite{lptype1, lptype2, lptypesurvey}.
\begin{definition}[LP-Type Problem]\label{def:lptype}
Let $S$ be a finite set and $\varphi$ be a function that maps subsets of $S$ to some value. We say $(S, \varphi)$ is a LP-Type problem if $\varphi$ satisfies the following two properties:
\begin{itemize}
    \item Monotonicity: if $A \subseteq B \subseteq S$ then $\varphi(A) \le \varphi(B)$
    \item Locality: For all $A \subseteq B \subseteq S$ and elements $x \in S$, if $\varphi(A) = \varphi(B) = \varphi(A \cup \{x\})$, then $\varphi(A) = \varphi(B \cup \{x\})$.
\end{itemize} 
\end{definition}

LP-Type problems have a natural definition of `dimension' which influences the runtime of many algorithms for LP-Type problems as well as our algorithm for property testing of LP-Type problems. First, we must define the notion of a basis.

\begin{definition}[Basis of LP-Type problems]\label{def:basis}
Given an LP-Type problem, a \emph{basis} $B \subseteq S$ is a set such that for all proper subsets $B' \subset B$, we have $\varphi(B') < \varphi(B)$.
\end{definition}

Given the above definitions, we can now define the dimension of a LP-Type problem.

\begin{definition}[Dimension of LP-Type problem]\label{def:dim}
The \emph{dimension} $\delta$ of an LP-Type problem is the largest possible size of a basis $B \subseteq S$. This is sometimes also called the \emph{combinatorial dimension}.
\end{definition}

An example of a LP-Type problem is when $S$ is a set of linear inequalities in $d$ dimensions and $\varphi$ is a linear objective function. In this case, the dimension of this LP-Type problem corresponds to the usual notion of dimension and is equal to $d$ \cite{lptype1, lptype2, Clarkson:1995:LVA:201019.201036}. There are many examples of well-studied LP-Type problems, and in many of these cases explicit bounds, if not exact values, are known regarding their dimensions. For more information, see our contributions in Section \ref{sec:contributions}. In general, the dimension of the problem tends to grow with the `difficulty' of solving it and for property testing, our query complexity bound is also a function of the dimension. We now formally define property testing for LP-Type problems.

\begin{definition}[Property Testing of LP-Type problems]\label{def:prop_testing}
Given an LP-Type problem $(S, \varphi)$, a parameter $k$, a distance parameter $\epsilon$, and query access to the constraints in $S$, we wish to distinguish the following two cases:
\begin{itemize}
    \item  Output \accept with probability at least $2/3$ if $\varphi(S) \le k$ (Completeness Case)
    \item Output \reject with probability at least $2/3$ if at least $\epsilon |S|$ constraints need to be removed from $S$ for $\varphi(S) \le k$ to hold (Soundness Case).
\end{itemize}
\end{definition}
\begin{remark}
 We say that $S$ is $\epsilon$-far if it falls in the soundness case.
 \end{remark}
\subsection{Our Contributions} \label{sec:contributions}
The main contribution of this paper is a comprehensive algorithm for property testing of LP-Type problems with query complexity $O(\delta/\epsilon)$ where $\delta$ is the dimension of the LP-Type problem. Note that this bound is independent of the number of constraints which is $|S|$. Our algorithm is simple and proceeds by first sampling a small set of random constraints in $S$, constructing a partial solution, and `testing' this partial solution against few other randomly chosen constraints. The analysis that we reject in the $\epsilon$-far case (soundness) is straightforward. However, the main technical challenge lies in showing that our algorithm accepts in the completeness case. To do so, we use a `sampling' lemma which roughly says that for a randomly chosen subset $R$ of $S$ of a particular size (depending on the dimension $\delta$) and $x$ a randomly chosen element of $S \setminus R$, we have $\varphi(R) = \varphi(R \cup \{x\})$. Using this result, we show that we are likely to accept in the completeness case. For the full detailed analysis, see Section \ref{sec:alg}.

We highlight the power of our approach by considering the query complexity bounds that we get for a few selected problems. In many cases, we are also able to get matching lower bounds. More specifically, we obtain the following results as an application of our framework:
\begin{enumerate}
    \item We consider the problem of determining if a set of linear inequalities in $d$ variables is feasible (there exists a satisfying assignment) or if at least $\epsilon$-fraction of the constraints need to be removed for the set of constraints to be feasible. While this problem does not exactly fall under the LP-Type definition (since there is no optimization function), we modify the `meta' algorithm slightly to given an algorithm with query complexity is $O(d/\epsilon)$. 
    
    \item We study the problem of determining if a set of points in $d$ dimensions labeled $\{ +1, -1\}$ is linearly separable or if at least $\epsilon$-fraction of the points need to be relabeled or removed for the points to be linearly separable. Using result $1$ above, we directly get a query complexity bound of $O(d/\epsilon)$. We also give a matching lower bound for this problem which implies a lower bound for result $1$.

    \item We obtain a result for property testing of many classical LP-Type problems. In particular, we consider the following problems:
    \begin{itemize}
        \item \emph{Smallest enclosing ball}: Accept if a set of points in $\mathbb{R}^d$ can be covered by a ball of radius $r$ and reject if at least $\epsilon$-fraction of the points need to be removed to be able to be covered by a ball of radius $r$. This problem has been previously studied in \cite{clustering}. We beat the upper bound obtained in this paper by getting a tight query complexity of $O(d/\epsilon)$ queries (also see point $4$ below).
        
        \item \emph{Smallest intersecting ball}: Accept if a set of closed convex bodies in $\mathbb{R}^d$ can all be intersected by a ball of radius $r$ and reject if at least $\epsilon$-fraction of the convex bodies need to be removed to be able to be intersected by a ball of radius $r$. 
        
        \item \emph{Smallest volume annulus}: Accept if a set of points in $\mathbb{R}^d$ can be enclosed in an annulus of volume $V$ and reject if at least $\epsilon$-fraction of the points need to be removed to be encloseable by an annulus of volume $V$. 
    \end{itemize}
    In all these cases, it is known that the dimension of the LP-Type problem is linearly related to the dimension of the points in $S$, so we get an upper bound of $O(d/\epsilon)$ queries.
    \item We get a matching lower bound of $O(d/\epsilon)$ queries for the smallest enclosing ball problem and the smallest intersecting ball problem. This provides a lower bound for the radius cost of clustering considered by Alon et al.\@ in \cite{clustering} in the case of $1$ cluster.
    \end{enumerate}
\begin{remark}
Note that there are also many examples of LP-Type problems where the constraints in $S$ describe points in dimension $d$ but the dimension of the LP-Type problem is not a linear function of $d$. For example, if $\varphi(S)$ is the smallest ellipsoid that encloses the set of points in $S$ which are in $\mathbb{R}^d$, then $(S,\varphi)$ has dimension $O(d^2)$ as a LP-Type problem \cite{lptypesurvey}. We did not explicitly highlight these problems but our approach also gives an upper bound on the query complexity for the property testing versions of these problems.
\end{remark}

\section{Meta Algorithm for Property Testing of LP-Type problems} \label{sec:alg}

We now present our `meta' algorithm, \textsc{LP-Type Tester}, for property testing of LP-Type problems as defined in Definition \ref{def:prop_testing}. Given a LP-Type problem $(S, \varphi)$, Our algorithm first samples a subset $R$ of  $O(\delta/\epsilon)$ constraints from $S$ where $\delta$ is the dimension of the LP-Type problem. It then calculates the value of $\varphi$ on the sampled subset. After this step, an additional $O(1/\epsilon)$ constraints are sampled randomly from $S$. If $\varphi(R \cup \{x\})$ differs from $\varphi(R)$ where $x$ is any of the additional random constraints, then our algorithm outputs \reject. Otherwise, the algorithm outputs \accept. We present our approach in Algorithm \ref{alg:lptype} along with our main theorem, Theorem \ref{thm:main} which proves the correctness of Algorithm \ref{alg:lptype}.

\begin{algorithm}[!ht]
	\SetKwInOut{Input}{Input}
	\SetKwInOut{Output}{Output}
	\Input{$\delta, \epsilon, k$, query access to constraints in $S$}
	\Output{\accept or \reject}
	\DontPrintSemicolon
	$r \gets \lceil 10\delta/\epsilon \rceil$ \; 
	$R \gets$ random sample of size $r$ of constraints from $S$. \;
	\If{$\varphi(R) > k$}{Output \reject and abort.}
	\For{$2/\epsilon$ rounds}{
	$x \gets$ uniformly random constraint of $S\setminus R$ \;
	\If{$\varphi(R \cup \{x\}) \ne \varphi(R)$}{
	Output \reject and abort.
	}
	}
	Output \accept.
	\caption{\textsc{LP-Type Tester}}
	\label{alg:lptype}
\end{algorithm}

\begin{theorem}[Correctness of \textsc{LP-Type Tester}]\label{thm:main}
Given an LP-Type problem $(S,\varphi)$ of dimension $\delta$ and parameters $k$ and $\epsilon$, the following statements hold with probability at least $2/3$:
\begin{itemize}
    \item \textbf{Completeness case:} \textsc{LP-Type Tester} outputs \accept $\varphi(S) \le k$.
    \item \textbf{Soundness Case:} \textsc{LP-Type Tester} outputs 
    \reject if at least $\epsilon |S|$ constraints need to be removed from $S$ for $\varphi(S) \le k$ to hold.
    \end{itemize}
\end{theorem}
\begin{remark}
Note that the query complexity of Algorithm \ref{alg:lptype} is $O(\delta/\epsilon)$ which is independent of $|S|$, the number of constraints. Furthermore, the runtime is polynomial in $\delta/\epsilon$ since we are only solving $O(1/\epsilon)$ linear programs in $\delta$ variables and $O(\delta/\epsilon)$ constraints.
\end{remark}

\paragraph{Overview of the proof}
To prove the correctness of \textsc{LP-Type Tester}, we analyze the completeness case and the soundness case separately. For the soundness case, we show that with sufficiently large probability, either $\varphi(R) > k$ or \textsc{LP-Type Tester} outputs \reject during the second sampling phase where we sample an additional $O(1/\epsilon)$ constraints. To show this, we use the locality property of LP-Type problems (see Definition \ref{def:lptype}) to show that there must be `many' $x$ such that $\varphi(R \cup \{x\}) \ne \varphi(R)$. To analyze the completeness case, we use the \emph{Sampling Lemma}, Lemma \ref{lem:sampling}, to show that there are `few' $x$ such that $\varphi(R \cup \{x\}) \ne \varphi(R)$ so that Algorithm \ref{alg:lptype} outputs \accept with sufficiently large probability.

Before we present the proof of Theorem \ref{thm:main}, we present the Sampling Lemma as described above. This lemma has previously been used to study LP-Type problems. For completeness, we present a proof. For more information, see \cite{lptype1, lptype2, violator1, samplinglem1}.  Before we present the lemma, we introduce two new definitions.

\begin{definition}[Violators and Extreme Elements] For a subset $R \subseteq S$, define the violators and extreme elements of $R$ as the following:
\begin{itemize}
    \item Define the \textbf{violators} of $R$ as the set $V(R) = \{s \in S \backslash R \mid \varphi(R \cup \{s\}) \ne \varphi(R)\}$.
    \item Define the \textbf{extreme elements} of $R$ as the set  $X(R) = \{s \in R \mid \varphi(R) \neq \varphi(R \backslash \{s\}\}$.
\end{itemize}
\end{definition}
\begin{remark}
Note that $s$ is a violator of $R$ if and only if $s$ is an extreme element of in $R \cup \{s\}$.
\end{remark}

We now present the Sampling Lemma.

\begin{lemma}[Sampling Lemma] \label{lem:sampling}
 Let $v_r = \E[|V(R)|]$ and $x_r = \E[|X(R)|]$ where both expectations are taken over the random subsets $R$ of $S$ which have size $r$. Suppose $|S| = n$. Then for $0 \le r \le n$,  we have
$$ \frac{v_r}{n-r} = \frac{x_{r+1}}{r + 1}.$$
\end{lemma}

\begin{proof}
Let $\mathbf{1}\{ \cdot \}$ denote an indicator variable. Note that
\begin{align*}
    \binom{n}{r}v_r &= \sum_{R \in \binom{S}{r}} \sum_{s \in S \backslash R} \mathbf{1}\{s \text{ is a violator of }R\} =\sum_{R \in \binom{S}{r}} \sum_{s \in S \backslash R} \mathbf{1}\{s \text{ is extreme for  } R \cup \{s\}\} \\
    &= \sum_{Q \in \binom{S}{r+1}} \sum_{s \in Q} \mathbf{1}\{s \text{ is extreme for } Q\} = \binom{n}{r+1}x_{r+1}.
\end{align*}
Since 
$$\frac{\binom{n}{r+1}}{\binom{n}{r}} = \frac{r! (n-r)!}{(r+1)!(n-r-1)!} = \frac{n-r}{r+1},$$ the proof is complete.
\end{proof}
\begin{remark}
Note that $(S, \varphi)$ does not need to be a LP-Type problem for the Sampling Lemma to hold true.
\end{remark}

If $(S, \varphi)$ is a LP-Type problem, there is a direct relationship between the expected number of violators and the dimension of $(S, \varphi)$ as defined in \ref{def:dim}. The following corollary also appears in many forms in literature (for instance \cite{lptype1, lptype2, violator1, violator2}) but we present its proof for completeness.

\begin{corollary}\label{cor:viol_bound}
Let $(S, \varphi)$ be a LP-Type problem of dimension $\delta$ and let $|S| = n$. If $R \subseteq S$ is subset of size $r$ chosen uniformly at random, then $v_r = \E[|V(R)|]$ satisfies
$$v_r \le \frac{\delta(n-r)}{r+1}.$$
\end{corollary}
\begin{proof}
We show that for any set $R \subseteq S$, we have $|X(R)| \le \delta$. Then the corollary follows from Lemma \ref{lem:sampling}. Let $R'$ be the smallest subset of $R$ such that $\varphi(R') = \varphi(R)$. We first claim that $V(R') = V(R)$. It is clear that $V(R) \subseteq V(R')$ by monotonicity (see Definition \ref{def:lptype}). For the other inclusion, consider $x \in V(R')$. If $x \not \in V(R)$, we have $\varphi(R \cup \{x\}) = \varphi(R) = \varphi(R')$ so by locality, we have $\varphi(R') = \varphi(R' \cup \{x\})$ which contradicts the fact that $x \in V(R')$. Therefore, our claim holds true.

We now claim that $R'$ is a basis as defined in Definition \ref{def:basis}. Suppose for the sake of contradiction that $R'$ is not a basis. Then there exists a $F \subset R'$ such that $\varphi(F) = \varphi(R')$. We now claim that $V(R') = V(F)$. It is clear that $V(R') \subseteq V(F)$. To show the other inclusion, let $x \in V(F)$. Then if $x$ was not a violator of $R'$, then $\varphi(R' \cup \{x\}) = \varphi(R') = \varphi(F)$ which would imply that $\varphi(F \cup \{x\}) = \varphi(F)$ by the locality property in Definition \ref{def:lptype} which is false by definition. Hence, $V(F) = V(R') = V(R)$ which contradicts the minimality of $R'$. Therefore, $R'$ is a basis. 

Finally, we claim that if $x \in X(R)$ then $x \in R'$. This must be true because otherwise, we have $R' \subseteq R \setminus \{x \} \subseteq R$ which results in a contradiction by monotonicity. Finally, since $X(R) \subseteq R'$ and $R'$ is a basis, it follows that $|X(R)| \le \delta$, as desired.
\end{proof}
\begin{remark}
Corollary \ref{cor:viol_bound} holds for a larger class of problems than LP-Type problems called \emph{violator spaces} (\cite{violator1, violator2}. However, we omitted this extra layer of abstraction since there are no additional natural property testing consequences from considering violator spaces over LP-Type problems.
\end{remark}

\begin{proof}[Proof of Theorem \ref{thm:main}]
We first prove the soundness case. Consider the set $R$ that was randomly sampled in step $2$ of \textsc{LP-Type Tester}. Assume that $\varphi(R) \le k$ since this can only decrease the probability that our algorithm outputs \reject. Now we claim that there must be at least $\epsilon |S|$ choices of $x$ in step $6$ of \textsc{LP-Type Tester} that results in $\varphi(R \cup \{x\}) > \varphi(R)$ (so that we correctly output \reject). To show this, note that if $\varphi(R) = \varphi(R \cup \{x\}) = \varphi(R \cup \{y\})$ for $x \ne y$ then by locality, it follows that $\varphi(R) = \varphi(R \cup \{x,y\})$. Therefore, if there are less than $\epsilon |S|$ choices of $x$ in step $6$ of \textsc{LP-Type Tester} such that $\varphi(R \cup \{x\}) > \varphi(R)$, then we would have $\varphi(R \cup R') = \varphi(R) \le k$ where $|R \cup R'| \ge (1-\epsilon)|S|$ which would contradict our assumption that at least $\epsilon|S|$ constraints need to be removed from $S$ for $\varphi(S) \le k$ to hold true. Therefore, the probability our algorithm does not output \reject in any of the $2/\epsilon$ rounds is at most
\begin{equation}\label{eq:completeness_bound}
     (1 - \epsilon)^{2/\epsilon} \le e^{-2} < \frac{1}{3} 
\end{equation}
which means that we output \reject with probability at least $2/3$, as desired.

We now prove Theorem \ref{thm:main} for the completeness case. Let $v_r = \E[|V(R)|]$. Since $r = |R| = 10\delta/\epsilon$, Corollary \ref{cor:viol_bound} gives us
$$ v_r \le \frac{\delta(|S|-r)}{r+1} \le \frac{\epsilon |S|}{10}.$$ Therefore in the completeness case, the probability that a randomly chosen $x$ satisfies $\varphi(R \cup \{x\}) \ne \varphi(R)$ is at most $\epsilon/10$. Since we choose $2/\epsilon$ random constraints, the probability we don't find such a $x$ is at least
\begin{equation}\label{eq:soundness_bound}
    (1 - \epsilon/10)^{2/\epsilon} \ge 1 - \frac{2}{10} > \frac{2}{3}.
\end{equation}
Therefore, \textsc{LP-Type Tester} outputs \accept with probability at least $2/3$, as desired.
\end{proof}

\section{Property Testing Applications of \textsc{LP-Type Tester}} \label{sec:applications}
We now give applications of the framework we build in Section \ref{sec:alg}. We first consider the problem of testing feasibility of a set of linear inequalities. As a direct consequence, we can test if a set of labeled points can be linearly sepearable (either by linear hyperplanes or by functions that have a finite basis). These two applications will not be an immediate corollary of Theorem \ref{thm:main} since there is no objective function that we want to optimize, but our results follow from Theorem \ref{thm:main} with some slight modificatons. 

We then consider direct applications of Algorithm \ref{alg:lptype} to some cannonical LP-Type problems such as the smallest enclosing ball. Theorem \ref{thm:main} gives direct upper bounds for property testing for these problems.

\subsection{Testing Feasibility of a System of Linear Equations}\label{sec:feasibility_ub}

We first begin by considering testing feasibility of a set of linear inequalities. Recall that in this problem, we have $n$ linear constraints in $\mathbb{R}^d$ (such as $x_1 + \cdots + x_d \le 1$) and we want to distinguish the following two cases with probability at least $2/3$:

\begin{itemize}
	\item The system of linear inequalities can all be mutually satisfied, i.e., the system is feasible (Completeness Case)
	\item At least $\epsilon |S|$ many of the constraints need to be removed (or flipped) for the system to be feasible (Soundness Case).
\end{itemize}
 This is not exactly a LP-Type problem since we do not have an optimization function $\varphi$. (Note that if $\varphi$ was an indicator function for a subset of constraints being feasible then $\varphi$ would break the locality condition in Definition \ref{def:lptype}). However, we perform a slight modification of Algorithm \ref{alg:lptype} to create a new algorithm for this problem.

Our algorithm for this testing problem, \textsc{Linear Feasibility Tester}, uses the fact that if we pick any arbitrary $x \in \mathbb{R}^d$, then $x$ will violate `many' of the linear constraints in $S$ in the completeness case. In the soundness case, we use ideas from \textsc{LP-Type Tester} and show that if we introduce an arbitrary linear optimization function (thus turning our problem into an instance of linear programming), then a solution that optimizes a small subset of the constraints will not violate `too many' of the other constraints. We present our algorithm below along with Theorem \ref{thm:feasibility} that proves its correctness.

\begin{algorithm}[!ht]
	\SetKwInOut{Input}{Input}
	\SetKwInOut{Output}{Output}
	\Input{$d, \epsilon$, query access to constraints of $S$}
	\Output{Accept or Reject}
	\DontPrintSemicolon
	$r \gets \lceil 10d/\epsilon \rceil$ \; 
	$R \gets$ random sample of size $r$ of constraints from $S$ \;
	Create the linear program $L$: $\max{x_1}$ subject to the constraints in $R$ \;
	$x \gets$ solution of $L$ \;
	\If{$L$ is not feasible }{
	Reject and abort
	}
	\For{$2/\epsilon$ rounds}{
	$y \gets$ uniformly random constraint of $S$ \;
	\If{$x$ does not satisfy $y$}{
	Output \reject and abort.
	}
	}
	Output \accept.
	\caption{\textsc{Linear Feasibility Tester}}
	\label{alg:feasibility}
\end{algorithm}
\begin{theorem}[Correctness of \textsc{Linear Feasibility Tester}]\label{thm:feasibility}
Given a set $S$ of linear inequalities in $\mathbb{R}^d$, 
the following statements hold with probability at least $2/3$:
\begin{itemize}
    \item \textbf{Completeness case:} \textsc{Linear Feasibility Tester} outputs \accept if there exists $x \in \mathbb{R}^d $ that satisfies all of the constraints in $S$.
    \item \textbf{Soundness Case:} \textsc{Linear Feasibility Tester} outputs 
    \reject if at least $\epsilon |S|$ constraints need to be removed from $S$ for $S$ to be feasible.
    \end{itemize}
\end{theorem}
\begin{remark}
Note that the query complexity of Algorithm \ref{alg:lptype} is $O(d/\epsilon)$ which is independent of $|S|$, the number of constraints. 
\end{remark}
\begin{proof}
 The proof of the soundness case follows similarly to Theorem \ref{thm:main} using the fact that for any $x$, there are at least $\epsilon |S|$ constraints in $x$ such that $x$ violates these constraints. Then the probability that \textsc{Linear Feasibility Tester} outputs \reject in this case can be calculated to be at least $2/3$ using the same bound as Eq.\@ \eqref{eq:completeness_bound} in the proof of Theorem \ref{thm:main}.
 
 For the completeness case, we note that if we introduce the optimization function $\varphi(S) = \max x_1$ subject to the constraints in $S$, then $(S, \varphi)$ is an LP-Type problem of dimension $d$ (assuming that the constraints are non degenerate which can be assumed by perturbing the constraints and then taking the limit of the perturbation to $0$. For more details, see \cite{Clarkson:1995:LVA:201019.201036, Seidel1991}). Now let $x$ be the solution to the linear program that we solved in Step $4$ of \textsc{Linear Feasibility Tester}. Using Corollary \ref{cor:viol_bound}, we know that if $|R| = 10\lceil d/\epsilon \rceil$, then the number of constraints $v_r$ in $S$ that satisfy $\varphi(R \cup \{y\}) \ne \varphi(R)$ is at most 
 $$ v_r \le \frac{d|S|}{r} \le \frac{\epsilon |S|}{10}$$
  in expectation. Knowing that $x$ not satisfying $y$ implies that $y$ is a violator of $R$, the probability that $x$ does not satisfy a randomly chosen $y$ is at most $\epsilon/10$. Thus, using the exact calculation as in Eq.\eqref{eq:soundness_bound} of Theorem \ref{thm:main}, we have that \textsc{Linear Feasibility Tester} outputs \accept in the completeness case with probability at least $2/3$, as desired.
\end{proof}

\subsubsection{Testing if Labeled Points can be Linearly Separated}
As a direct consequence of the Theorem \ref{thm:feasibility}, we can test if a set of points in $d$ dimensions labeled $\{+1, -1\}$ can be linearly separated. More formally, we have the following corollary.
\begin{corollary}\label{cor:linear}
Given a set $S$ of points in $\mathbb{R}^d$ with labels in $\{+1, -1\}$, the following statements hold with probability at least $2/3$:
\begin{itemize}
    \item \textbf{Completeness case:} \textsc{Linear Feasibility Tester} outputs \accept if there exists a hyperplane that separates the two sets of labeled points.
    \item \textbf{Soundness Case:} \textsc{Linear Feasibility Tester} outputs 
    \reject if at least $\epsilon |S|$ points need to be removed (or relabeled) for $S$ to be linearly sepearable. 
    \end{itemize}
\end{corollary}
\begin{proof}
 The proof follows directly from the fact that we can write a linear inequality that represents a separating hyperplane. For example, if $p \in S$ is labeled $1$, we want to find $x$ such that $p^Tx \ge 1$ and if $p$ is labeled $-1$, we want to find $x$ such that $p^Tx \le -1$. 
\end{proof}

We consider generalizations of this problem where we wish to separate labelled points by arbitrary functions, rather than just linear hyperplanes, and where points can have multiple labels.
\paragraph{Separating labeled points using arbitrary functions:}
We can generalize our result by separating labeled points using arbitrary \emph{functions}: given a family of functions $\mathcal{F}$, we can ask if there is a $f \in \mathcal{F}$ such that $f(p) > 0$ for all points with a particular label and $f(p) < 0$ for all the points with the other label. 

We now translate this problem to a setting with linear inequalities. Our approach is standard in machine learning and is known as feature maps. If the family $\mathcal{F}$ has a finite basis $f_1, \cdots, f_k$, meaning that every $f \in \mathcal{F}$ is a linear combination of $f_1, \cdots, f_k$, then we can create a system of linear inequalities as follows. For each point $p \in S$, we can make a new constraint which is $(f_1(p), \cdots, f_k(p)) x \ge 1$ (note there that $x$ is a column vector of variables) if $p$ has one particular label or $\le -1$ if $p$ has another label. Then this system of linear constraints is feasible iff there are scalars $a_1 , \cdots, a_k$ such that $\sum_i a_i f_i(p) \ge 0$ for all $p$ with one label and $\sum_i a_i f_i(p) \le 0$ for all $p$ with the other label. Then our separating function is precisely $f = \sum_i a_i f_i$. Note that in this formulation, we have $k$ variables. Thus, the query complexity is $O(k/\epsilon)$.

As an example, we consider the case that $\mathcal{F}$ is the family of polynomials in $d$ variables with degree $\le t$. The basis of this family is all the possible terms of the form $x_1^{t_1} \cdots x_d^{t_d}$ where the $t_i$ are non-negative and add to at most $t$. By a standard balls and bins argument, the number of these terms is $ \binom{t+d}d$. For constant $t$, this is $O(d^t)$, which means that our system of linear constraints has $O(d^t)$ variables. Thus, the query complexity is $O(d^t/\epsilon)$.

\paragraph{Separating points with multiple labels}
Suppose that instead of assigning each point one of $2$ labels, we instead chose to assign it one of $\ell \ge 2$ labels. One common interpretation of separability for this setup is to check if each of the $\binom{\ell}2$ pairs of label sets are separable. We modify our notion of $\epsilon$-far to reflect this.

\begin{definition}
 $S$ is $\eps$-far from linearly separable if at least $\eps |S|$ many labels in $S$ have to be changed for $S$ to be separable.
\end{definition}

If such a data set is $\eps$-far from separable, then some subset with consisting of two labels must be $\eps/\binom{\ell}2$-far from separable. As such, we can consider an algorithm that runs Algorithm \textsc{Linear Feasibility Tester} on each pair of labels with $\eps' = \eps/\binom{\ell}2$ and outputs \accept if all these tests output \accept. We need to reduce the error probability for each pair such that the overall error probability of outputting the incorrect answer (acquired by a union bound) is still at most $1/3$. This can be done by using a stronger version of the original algorithm where we run it $108\log \ell = O(\log \ell)$ times and taking the majority answer. By a standard Chernoff bound argument, the probability this process gives the wrong answer is at most $e^{-(\frac{1}{2})^2\frac{1}{3}*108\log \ell / 3} \le 1/\ell^3$. Thus, we can distinguish separability in this case by running this stronger version over all pairs of distinct labels, resulting in $O(\ell^2 \log \ell)$ instances of \textsc{Linear Feasibility Tester}, using $\eps' = \eps/\binom{\ell}2$. So, the total query complexity will be $O(d\ell^4\log \ell/\epsilon)$.

Additionally, the completeness case has error at most $\binom{l}{2} * 1/\ell^3 = o(1)$ by a Union Bound argument. Clearly, the soundness case has error at most $1/\ell^3$, since there is one pair of distinct labels which is $\eps'$-far from separable.

\subsection{Upper Bounds for Canonical LP-Type Problems}
We now give direct applications of \textsc{LP-Type Tester} to some canonical LP-Type problems. The correctness of these applications follows directly from Theorem \ref{thm:main}. Our list is not exhaustive and we only consider some of the more well known LP-Type problems. In all of the following problems, Theorem \ref{thm:main} tells us that the following statements hold with probability at least $2/3$:
\begin{itemize}
    \item \textsc{LP-Type Tester}  outputs \accept if $\varphi(S) \le k$ (Completeness Case)
    \item  \textsc{LP-Type Tester} outputs 
    \reject if at least $\epsilon |S|$ constraints need to be removed from $S$ for $\varphi(S) \le k$ to hold (Soundness Case).
    \end{itemize}

Our results are the following:
\begin{itemize}
    \item \textit{Smallest enclosing ball}: In this problem, $\varphi(S)$ is the radius of the smallest enclosing ball of a set of points $S$ in $\mathbb{R}^d$. It is known that the dimension of this LP-Type problem is $d+1$ (see \cite{lptypesurvey}) so we can test if $\varphi(S) \le k$ with query complexity $O(d/\epsilon)$ queries.
    \item \emph{Smallest intersecting ball}: In this problem, $\varphi(S)$ is the smallest radius ball that intersects a set of closed convex bodies $S$ in $\mathbb{R}^d$. The dimension of this LP-Type problem is $O(d)$ (\cite{lptypesurvey}) so we can test if $\varphi(S) \le k$ with query complexity $O(d/\epsilon)$ queries.
    
    \item \emph{Smallest volume annulus}: In this problem, $\varphi(S)$ is the volume of the smallest annulus that contains a set of points $S$ in $\mathbb{R}^d$. Again, the dimension of this LP-Type problem is $O(d)$ ( \cite{lptypesurvey}) so we can test if $\varphi(S) \le k$ with query complexity $O(d/\epsilon)$.
\end{itemize}

\section{Lower Bounds}\label{sec:lowerbounds}
In this section, we give matching lower bounds for all the testing problems that we considered in Section \ref{sec:applications}.

\subsection{Lower Bound for Testing Feasibility of Linear Constraints}
Since linear separability is a special case of feasibility of linear constraints, we can lower bound the necessary query complexity of the latter by providing one for the former. In particular, we aim to show that $\Omega(d/\epsilon)$ queries are needed to determine if a set of points in $d$ dimensions is linearly separable. By the reduction of linear separability to feasibility of linear constraints, this implies that $\Omega(d/\epsilon)$ constraint queries are needed to test feasibility of a system of linear constraints, which matches our upper bound.

Our overall approach is to first introduce a set of $O(d)$ points in $\mathbb{R}^{d}$ that have the property that if we do not look at a large enough collection of these points, they can be separated by a hyperplane even with arbitrary labels. However, there will exist a labeling of all of the points such that `many' of the points will have to be removed or relabeled for this labeling to be separated. The existence of these points is given in Lemma \ref{lem:points} (and is inspired by the moment curve).

Then, repeating these points with carefully chosen multiplicities allows us to construct our set $S$ of points. Then a coupon collector argument gives us our desired lower bound on the query complexity. This argument is formalized in the proof of Theorem \ref{thm:linseplb}.

\begin{lemma}\label{lem:points}
There exists a set $S$ of $3d+1$ points in $\mathbb{R}^d$ that satisfy the following conditions:
\begin{enumerate}
    \item There exists a labeling of the points of $S$ such that at least $d$ points have to be relabeled for the points to be linearly separable.
    \item Any subset of points of $S$ of size $d+1$ with arbitrary labels in $\{-1, 1\}$ is linearly separable.
\end{enumerate}
\end{lemma}
\begin{proof}
We construct our set $S$ as follows. Let $x_i$ be the point $$(i^1, \cdots, i^d) \in \mathbb{R}^d$$ for $1\le i \le 3d+1$ (note that this set of points is referred to as the moment curve). We prove the first claim using a standard relationship between the moment curve and polynomials. Assign the point $x_i$ to the label $(-1)^i$. 
Let $k$ be the number of relabeled points such that $S$ is linearly separable. Then there exists $w \in \mathbb{R}^d$ and $w_0 \in \mathbb{R}$ such that $\text{Sign}({x_i^Tw+w_0})$ matches the label of every point $x_i \in S$. In other words, there exists a polynomial 
\begin{equation} \label{eq:polyform}
    P(x) = \sum_{j=0}^d c_jx^j 
\end{equation}
such that $\text{Sign}({P(i)})$ matches the label of $x_i$. Now note that if there are two consecutive indices $i$ and $i+1$ that have different labels, then $P$ must have a root in the interval $(i,i+1)$. Originally, there are $3d$ such alternating intervals. Now note that the relabeling of any point can decrease the total number of such alternating intervals by at most $2$. Hence after $k$ relabelings, there must be at least $3d-2k$ alternating intervals. However, since $P$ is a $d$ degree polynomial, it must have at most $d$ roots which means $3d-2k \le d$ and therefore, $k \ge d$, as desired. 

We now prove the second claim. Let $x_{a_1}, \cdots, x_{a_{d+1}}$ be a subset of $d+1$ points of $S$. Without loss of generality, suppose that $a_1 < \cdots < a_{d+1}$. We now show that for every labelings of these $d+1$ points, there exists a polynomial of degree $d$ such that the sign of $P(a_i)$ matches the label of $x_{a_i}$. Towards this goal, pick $t$ elements $b_1, \cdots, b_t$ of the set $\{a_2, \cdots, a_{d+1} \}$ where $t \le d$. Consider the $t+1$ intervals
$$[a_1, b_1), [b_1, b_2), \cdots, [b_{t-1}, b_t), [b_t, a_{d+1}+1).$$
We can then find a polynomial of degree $d$ such that 
\begin{itemize}
    \item the sign of $P$ is constant on $I \cap \{a_1, a_2, \cdots, a_{d+1} \}$ where $I$ is any of the $t+1$ intervals above,
    \item the sign of $P$ alternates between consecutive intervals.
\end{itemize}
This is possible since we are only specifying the value of $P$ on $d+1$ locations. Now the total number of labelings described by all possible choices of $P$ is given by 
$$2 \sum_{t = 0}^d \binom{d}t = 2^{d+1} $$
where the factor of $2$ comes from specifying the sign of $P$ on the first interval. Note that $2^{d+1}$ is exactly the total number of different ways to label $d+1$ points, which proves the second claim.
\end{proof}
With Lemma \ref{lem:points} on hand, we can prove our desired lower bound on the query complexity.

\begin{theorem}\label{thm:linseplb}
Any algorithm that tests if a set $S$ of labeled points in $d$ dimensions can be linearly separated requires $\Omega(d/\epsilon)$ queries.
\end{theorem}
\begin{proof}
Let $|S| = n$. We create two families of $n$ points in $\mathbb{R}^d$ with a specific labeling such that any $S$ from one family can be linearly separated while any $S$ from the other family is $\epsilon$-far from being linearly separable. 
First, consider the set of $3d+1$ points supplied by Lemma \ref{lem:points} and the labeling from part $1$ of the lemma. The first family $\mathcal{F}_1$ consists of picking a subset of $d+1$ of these points (with the labeling above), repeating $d$ of these points $n\epsilon/d$ times, and repeating the remaining point $(1-\epsilon)n$ times. The second family $\mathcal{F}_2$ (again with the same labeling) consists of picking all of the $3d+1$ points from Lemma \ref{lem:points}, repeating some $3d$ of these points with multiplicity $n\epsilon/(3d)$, and repeating the last point with multiplicity $(1-\epsilon)n$. 

By Lemma \ref{lem:points}, we know that if $S$ is from $\mathcal{F}_1$ then $S$ is linearly separable while if $S$ is from $\mathcal{F}_2$, then $S$ is at least $\epsilon/(3d) \cdot d = O(\epsilon)$-far from separable. Any algorithm that queries points randomly must discover at least $d+1$ \emph{unique} points out of the points that were repeated $n\epsilon/d$ time from any $S$ in $\mathcal{F}_2$
to discover that this $S$ is $O(\epsilon)$-far from separable (otherwise, the points look separable). Call points that are identical `groups'. Now given a random point from $S$, the probability of hitting any one group is $\epsilon/(3d)$. Therefore by coupon collector, the expected number of queries required to hit at least $d+1$ of these $3d$ groups is at least
$$ \frac{1}{\epsilon} \left( \frac{3d}{3d} + \frac{3d}{3d-1} + \cdots + \frac{3d}{3d-d} \right) = \frac{3d}{\epsilon}( H_{3d} - H_{2d-1} ) = \Theta \left( \frac{d}{\epsilon} \right), $$
as desired.
\end{proof}
As a corollary, we have the following lower bound as well. This is due to the reduction from linear separability to linear program feasibility from the proof of Corollary \ref{cor:linear}.

\begin{theorem}
Any algorithm that tests if $n$ linear inequalities in $d$ dimensions are feasible requires $\Omega(d/\epsilon)$ queries. 
\end{theorem}

We now give matching query complexity lower bounds for the LP-Type problems that we considered in Section \ref{sec:applications}. 
\subsection{Lower bound for Testing Smallest Enclosing Ball}
We first give a lower bound for property testing the radius of the smallest enclosing ball of a set of points. Our approach is to first construct a set of points in $\mathbb{R}^j$, for any $j$, whose smallest enclosing ball can be calculated exactly. This set of points will have the property that a small enough subset of the points will have a significantly smaller enclosing ball. Therefore, if an algorithm does not query enough points, it will incorrectly believe that this set of points can be covered by a ball of small radius. Our construction for this case will be a regular simplex and explained below.

\begin{lemma}\label{lem:radius}
The radius of the circumcircle of a unit simplex in $\mathbb{R}^j$ is $\sqrt{j}/({ \sqrt{2(j+1)}})$.
\end{lemma}
\begin{proof}
 Note that we can embed a regular $j$-simplex in $\mathbb{R}^{j+1}$ using the coordinates $\{e_i\}_{i=1}^{j+1}$ where $e_i$ is the all zero vector with a single $1$ in the $i$th coordinate. This simplex has edge length $\sqrt{2}$ so we can scale appropriately to find the circumcircle of a unit simplex. Now the centroid of this simplex is easily seen to be located at $(1/(j+1), \cdots, 1/(j+1))$ which means that the circumcircle has radius
 $$ \sqrt{ \left(1 - \frac{1}{j+1} \right)^2 + \frac{j}{(j+1)^2}} = \sqrt{\frac{j}{j+1}}.$$
 Now scaling by $1/\sqrt{2}$ gives us the desired value.
\end{proof}

\begin{theorem}\label{thm:ballslb}
Any algorithm that tests if a set of $n$ points in $\mathbb{R}^d$ can be enclosed by a ball of radius $k$, where $k$ is given, requires $\Omega(d/\epsilon)$ queries.
\end{theorem}
\begin{proof}
Let $k$ be fixed. We construct two families of points in $\mathbb{R}^{O(d)}$ such that any $S$ from one family can be enclosed by a ball of radius $k$ while any $S$ from the second family is $\epsilon$-far from being enclosed by a ball of radius $k$.
Before constructing these families, we first pick $\ell$ such that the regular simplex of side length $\ell$ in $\mathbb{R}^{d+1}$ has circumradius $k$. 

Now to create the first family $\mathcal{F}_1$, we first pick any $d+1$ points of the regular simplex with side length $\ell$ in $\mathbb{R}^{3d+1}$. Then we repeat one of these points with multiplicity $(1-\epsilon)n$ and we repeat the other $d$ points with multiplicity $n\epsilon/d$ each. To create the second family $\mathcal{F}_2$, we pick a point of the regular simplex with side length $\ell$ in $\mathbb{R}^{3d+1}$, repeat it with multiplicity $(1-\epsilon)n$, and repeat the other $3d$ points with multiplicity $n \epsilon/(3d)$. Finally, let $S$ be a set of $n$ points from $\mathcal{F}_2$. We can check that the circumradius of a regular unit simplex is an increasing function of the dimension and that any subset of the vertices of a regular simplex is a regular simplex itself. Therefore, the smallest radius of the points in $S$ is much larger than $k$ and $S$ is $O(\epsilon)$-far from being encloseable by a ball of radius $k$. However, similar to the argument in Theorem \ref{thm:linseplb}, any algorithm that rejects $S$ must have discovered at least $d+1$ distinct `groups' of repeated points. By the same coupon collector argument as in the proof of Theorem \ref{thm:linseplb}, we have that this task takes at least $\Omega(d/\epsilon)$ queries in expectation.
\end{proof}

As a simple application of Theorem \ref{thm:ballslb}, we get the following lower bounds as well.

\begin{corollary}
Any algorithm for testing the smallest intersecting ball for $n$ convex bodies in $\mathbb{R}^d$ requires $\Omega(d/\epsilon)$ queries.
\end{corollary}
\begin{proof}
The proof follows from the fact that a set of singleton points is also a set of convex bodies. In this case, the smallest intersecting ball is equivalent to the smallest ball that encloses these points. Therefore, the same lower bound as in Theorem \ref{thm:ballslb} holds. 
\end{proof}

\appendix
\section*{Appendices}

\section{Tolerant Tester for Testing Feasibility of Linear Constraints}
We generalize our argument in Section \ref{sec:feasibility_ub} by giving a \emph{tolerant} tester for testing feasibility of a system of linear constraints. In the tolerant version, we output \accept if there only `few' constraints need to be removed for a set of linear inequalities to be feasible. More formally, we wish to distinguish the following two cases with probability at least $2/3$:

\begin{itemize}
	\item At most $c\epsilon |S|$ many inequalities in $S$ need to be removed for $S$ (or flipped) for $S$ to be feasible, i.e., $S$ is $c\epsilon$-close to being feasible for some fixed $c < 1$ (Completeness Case).
	\item At least $\epsilon |S|$ many of the constraints need to be removed (or flipped) for the system to be feasible (Soundness Case).
\end{itemize}

Our approach is a slightly modified version of \textsc{Linear Feasibility Tester}, Algorithm \ref{alg:feasibility}, that we presented in Section \ref{sec:feasibility_ub}. The challenge here is the completeness case where we must accept if we only have a `few' bad constraints. To accomplish this, we carefully select a solution to a small linear program that we run. For more details, see Algorithm \ref{alg:lp_tolerant}. Our main theorem in this section, Theorem \ref{thm:tolerant} shows that we can perform tolerant testing using the same query complexity we used for the one-sided tester in Section \ref{sec:feasibility_ub}, namely $O(d/\epsilon)$. However, as we will explain below, the \emph{running time} of Algorithm \ref{alg:lp_tolerant}, \textsc{Tolerant Linear Feasibility Tester}, is exponential in the running time of Algorithm \ref{alg:feasibility}. Our algorithm, \textsc{Tolerant Linear Feasibility Tester}, is presented below.

\begin{algorithm}[!ht]
	\SetKwInOut{Input}{Input}
	\SetKwInOut{Output}{Output}
	\Input{$d, \epsilon$, query access to constraints of LP}
	\Output{Accept or Reject}
	\DontPrintSemicolon
	$r \gets \lceil 10d/\epsilon \rceil$ \; 
	$R \gets$ random sample of size $r$ of constraints from $S$ \;
	$x \gets$ solution of the largest subset $R'$ of $R$ such that the linear program $L$: $\max{x_1}$ subject to the constraints in $R'$  is feasible\;
	\If{No $L$ is not feasible }{
	Reject and abort
	}
	\For{$2/\epsilon$ rounds}{
	$y \gets$ uniformly random constraint of $S$ \;
	\If{$x$ does not satisfy $y$}{
	Output \reject and abort.
	}
	}
	Output \accept.
	\caption{\textsc{Tolerant Linear Feasibility Tester}}
	\label{alg:lp_tolerant}
\end{algorithm}
Unlike \textsc{Linear Feasibility Tester} where we run a linear program, we solve a slightly different program given in step $3$ of \textsc{Tolerant Linear Feasibility Tester}. The step determines the largest feasible subset of these constraints. Note that this step is clearly exponential in the number of constraints (which is $O(d/\epsilon)$). Therefore, the overall \emph{runtime} of \textsc{Tolerant Linear Feasibility Tester} will be exponential in the runtime of \textsc{Linear Feasibility Tester}. The correctness of \textsc{Tolerant Linear Feasibility Tester} is proven in Theorem \ref{thm:tolerant}.

\begin{theorem}[Correctness of \textsc{Tolerant Linear Feasibility Tester}]\label{thm:tolerant}
Given a set $S$ of linear inequalities in $\mathbb{R}^d$, there exists a constant $c < 1$ such that the following statements hold with probability at least $2/3$:
\begin{itemize}
    \item \textbf{Completeness case:} \textsc{Tolerant Linear Feasibility Tester} outputs \accept if there exists $x \in \mathbb{R}^d $ that satisfies $(1-c\epsilon)|S|$ of the constraints in $S$.
    \item \textbf{Soundness Case:} \textsc{Tolerant Linear Feasibility Tester} outputs 
    \reject if at least $\epsilon |S|$ constraints need to be removed from $S$ for $S$ to be feasible.
    \end{itemize}
\end{theorem}
\begin{remark}
Note that the query complexity of Algorithm \ref{alg:lp_tolerant} is $O(d/\epsilon)$ which is independent of $|S|$, the number of constraints.
\end{remark}
\begin{proof}

Note that the proof of the soundness case is identical to the proof of the soundness case in Theorem \ref{thm:feasibility} since for any $x$ we find in step $3$ of \textsc{Tolerant Linear Feasibility Tester}, there exists at least $\epsilon |S|$ choices of $y$ in step $7$ such that $x$ does not satisfy the constraint $y$. Then a similar calculation as in Eq.\@ \eqref{eq:completeness_bound} implies that we reject with probability at least $2/3$. 

We now focus on the completeness case where we know there is a subset of $(1-c \epsilon)|S|$ constraints that are feasible. We call this the \emph{good} set, and the rest, the \emph{bad} set. Consider the sample $R$ from step $2$ of \textsc{Tolerant Linear Feasibility Tester}. The expected number of constraints from the good set in $R$ is $(1-c\epsilon)r$. This means at most $c\epsilon r$ constraints in $R$ come from the bad set in expectation. Hence with probability at least $9/10$, we know that the number of constraints from the bad set is at most $10c\epsilon r$ by Markov's inequality, which means the number of constraints coming from the good set is at least $(1-10c\epsilon)r$. We condition on this event. Now note that one valid subset $R'$ to use in step $3$ of \textsc{Tolerant Linear Feasibility Tester} is to take all the constraints coming from the good set only. This results in $|R'| \ge (1-10c\epsilon)$. Since we are maximizing $|R'|$, this means that at most $10c\epsilon r$ of the constraints coming from the good set that are in $R$ will not be included in $R'$. Thus, $x$ satisfies at least $(1-20c\epsilon)r$ constraints in the good set with probability at least $9/10$. Now we proceed similarly as the proof of Theorem \ref{thm:feasibility}. By Corollary \ref{cor:viol_bound}, the probability that $x$ violates any other constraint in the good set is at most
$$ \frac{d(n'-r+1)}{n'(r-d)} \le \frac{dn'}{10dn'/\epsilon} = \frac{\epsilon}{10}$$ 
where $n'$ is the size of the good set. Furthermore, $x$ can possibly violate any constraint in the bad set which means that the probability $x$ violates any other constraint is at most $\epsilon/10 + c \epsilon < \epsilon/6$ for sufficiently small $c$. Then, the probability that we find such a constraint in $2/\epsilon$ rounds is at most
$$1 - \left(1 - \frac{\epsilon}{6}\right)^{2/\epsilon} \le 1 - \left(1 - \frac{1}{3} \right) = \frac{1}{3}.$$ Therefore, we accept with probability at least $2/3$, as desired. Note that we can take any $c < 1/15$ in the statement of the Theorem for instance. 
\end{proof}

\paragraph{Acknowledgements}
We would like to thank Ronitt Rubinfeld, Piotr Indyk, Bertie Ancona, and Rikhav Shah for helpful feedback.
\bibliographystyle{alpha}
\bibliography{Paper}

\end{document}